\documentclass[preprint,12pt,authoryear]{elsarticle}

\usepackage{amssymb}
\usepackage{graphicx}
\usepackage{braket}
\usepackage{amsmath}
\usepackage{float}
\usepackage{amsthm}
\usepackage{longtable}
\usepackage{booktabs}
\usepackage{CJKutf8}
\usepackage[strings]{underscore}
\usepackage{braket}
\usepackage{enumitem}
\usepackage{diagbox}
\usepackage{mathrsfs}
\usepackage{tikz}
\usetikzlibrary{graphs, graphs.standard, positioning, arrows.meta}

\usepackage[bottom]{footmisc}

\newtheorem{theorem}{Theorem}

\begin{document}
\begin{CJK}{UTF8}{gbsn}
\begin{frontmatter}



\title{Application of the Quantum Approximate Optimization Algorithm in Solving the Total Domination Problem}


\author[a,b]{Haoqian Pan\corref{cor1}}
\ead{52215500040@stu.ecnu.edu.cn}
\author[b]{Hang Yuan}
\author[b]{Yang Liu}
\author[a]{Changhong Lu}
\author[c]{Wanfang Chen}
\author[a]{Shiyue Wang}
\affiliation[a]{organization={School of Mathematical Sciences,  Key Laboratory of MEA(Ministry of Education) \& Shanghai Key Laboratory of PMMP, East China Normal University}, 
                        city={Shanghai},
                        postcode={200241},
                        country={China}}
\affiliation[b]{
  organization={China Communications Information \& Technology Group Co., Ltd.},
  city={Beijing},
  postcode={101399},
  country={China}
 }   
 \affiliation[c]{organization={School of Mathematical Sciences, University of Science and Technology of China}, 
                        city={Hefei},
                        postcode={230026},
                        country={China}}

\cortext[cor1]{Corresponding author.}

\begin{abstract}
 Recent advancements in quantum computing have spurred substantial research into the application of quantum algorithms to combinatorial optimization problems. Among these challenges, the Total Domination Problem (TDP) emerges as a classic and critical paradigm in the field. For a graph \( G(V, E) \), TDP entails finding a minimal subset \( D \subseteq V \) that contains no isolated vertices, where every vertex not in \( D \) has at least one neighbor in \( D \). TDP finds extensive applications across domains such as computer networks, social networks, and communications. Since the latter half of the last century, research efforts have focused on establishing its NP-completeness and developing solution algorithms, which have become foundational to combinatorial mathematics. Despite this rich history, the application of quantum algorithms to TDP remains largely underexplored. In this study, we present a pioneering application of the Quantum Approximate Optimization Algorithm (QAOA) to tackle TDP, evaluating its efficacy across a diverse set of parameters. This paper proves that the upper bound on the number of qubits required to solve TDP is \( 2|V| + |V| \log_{2}\left( \frac{2|E|}{|V|} - 1 \right) \). Our experimental findings demonstrate that QAOA is effective in addressing TDP: under most parameter combinations, it successfully computes a valid total dominating set (TDS). However, the algorithm’s performance in identifying the optimal TDS is contingent upon specific parameter choices, revealing a significant bias in the distribution of effective parameter points. This research contributes valuable insights into the potential of quantum algorithms for solving TDP and lays a solid groundwork for future investigations in this area.  
\end{abstract}

\begin{keyword}
    Quantum approximate optimization algorithm \sep Total domination problem \sep Combinatorial optimization
\end{keyword}

\end{frontmatter}


\section{Introduction}\label{sec:Introduction}

In recent years, advancements in quantum computing \citep{RN421, RN439, RN459} have driven the development of diverse quantum algorithms, including the Quantum Approximate Optimization Algorithm (QAOA) \citep{RN436}, the Quantum Adiabatic Algorithm \citep{farhi2001quantum}, and Quantum Annealing \citep{kadowaki1998quantum}. This progress has sparked significant interest in applying quantum algorithms to combinatorial optimization problems, such as the Max Clique Problem \citep{RN462}, Max Cut Problem \citep{RN436}, Max Flow Problem \citep{RN420}, and Minimal Vertex Cover Problem \citep{RN334}. Despite the growing momentum around quantum algorithms in combinatorial optimization, a notable research gap persists for specific problems—particularly the Total Domination Problem (TDP), a variant of the Domination Problem (DP). Prior to introducing TDP, we first present a formal definition of the dominating set (DS). In graph theory, a DS for a graph \( G(V, E) \) is defined as a subset \( D \subseteq V \) where every vertex \( v \in V \setminus D \) has at least one neighbor in \( D \). A total dominating set (TDS) imposes the additional constraint that no isolated vertices exist within \( D \). Examples of a DS and a TDS are provided in Fig. \ref{fig:ds} and Fig. \ref{fig:tds}, respectively.   

\begin{figure}[htbp]
  \centering
  \begin{tikzpicture}[scale=1.2,
    every node/.style={circle, draw, minimum size=1em, inner sep=1pt},
    every edge/.style={draw, thick}]
    \node(0) at (1,1) {0};
    \node[red, fill=red!20](1) at (2,1) {1};
    \node(2) at (3,1) {2};
    \node[red, fill=red!20](3) at (4,1) {3};
    \draw (0) edge (1)
          (1) edge (2)
          (2) edge (3);
  \end{tikzpicture}
  \caption {A DS $\{1,3\}$}
  \label{fig:ds}
\end{figure}

\begin{figure}[htbp]
  \centering
  \begin{tikzpicture}[scale=1.2,
    every node/.style={circle, draw, minimum size=1em, inner sep=1pt},
    every edge/.style={draw, thick}]
    \node(0) at (1,1) {0};
    \node[red, fill=red!20](1) at (2,1) {1};
    \node[red, fill=red!20](2) at (3,1) {2};
    \node(3) at (4,1) {3};
    \draw (0) edge (1)
          (1) edge (2)
          (2) edge (3);
  \end{tikzpicture}
  \caption {A TDS $\{1,2\}$}
  \label{fig:tds}
\end{figure}

The objectives of DP and TDP are to find the smallest DS and TDS. Since the last century, significant strides have been made in studying TDP, with comprehensive reviews available in \citep{RN427,RN473,RN474,RN481}. TDP has been established as NP-complete for various types of graphs, including general graphs, bipartite graphs, and comparability graphs \citep{RN463}, as well as split graphs \citep{RN464}, chordal graphs \citep{RN465}, circular graphs \citep{RN467}, line graphs, and line graphs of bipartite and claw-free graphs \citep{RN466}. Research into the NP-completeness of TDP has led to the development of algorithms tailored to specific graph classes, often aimed at minimizing time complexity. For instance, studies have investigated interval graphs \citep{RN477,RN470}, circular-arc graphs \citep{RN471,RN472}, cocomparability graphs \citep{RN477,RN468}, asteroidal triple-free graphs \citep{RN478}, distance hereditary graphs \citep{RN479}, and tree graphs \citep{RN474}. To date, however, the literature lacks investigations into the application of quantum computing techniques to the TDP. Although a handful of studies have investigated quantum algorithms for classical DP, for instance, \citet{RN415} solved the classical DP via indirect modeling followed by result correction, and \citet{zhang2024quantum} addressed the DP using Grover’s algorithm. Research on quantum algorithms for the TDP remains unexplored. The potential of quantum algorithms to effectively solve TDP, as well as their performance metrics, remains an open question for future research.

The main contributions of this paper are as follows:  

(1) This paper presents, for the first time, the Hamiltonian formulation of TDP. It provides and proves the strict upper bound on the number of qubits required to solve TDP, and demonstrates the reduction in the number of qubits compared with classical DP. For the first time, it is mathematically proven that TDP is simpler than DP at the quantum computing level.  

(2) This paper presents the first application of QAOA to solve TDP on a quantum simulator. We validated the effectiveness of QAOA across 128 parameter combinations. The experimental results show that QAOA successfully computes valid TDS for 93 of these combinations, with 12 combinations yielding the optimal TDS. Additionally, we found that under certain parameter settings, the accuracy of the final sampling results from QAOA can reach approximately 70\%, while the optimal probability can approach 25\%.  

Based on the above contributions, we affirm that utilizing QAOA to solve TDP holds significant potential and merits further investigation. This work represents one of the earliest contributions to the application of quantum computing in addressing TDP, and our findings will provide valuable insights for subsequent research employing quantum algorithms for this problem.  

The structure of this paper is as follows. In Section \ref{sec:problemmodeling}, we analyze TDP and present its 0-1 integer programming model, gradually deriving the process of converting it into a Hamiltonian. In Section \ref{sec:QAOA}, we outline the basic workflow of QAOA. Section \ref{sec:Experiment} details our experiments using a quantum simulator on a 6-node graph, comparing the performance of QAOA in solving TDP across different parameter combinations. Finally, Section \ref{sec:conclusion} provides a summary of the entire paper. 

\section{Problem modeling}\label{sec:problemmodeling}

This paper primarily investigates the use of QAOA to solve TDP. A key challenge in applying quantum algorithms to combinatorial optimization problems is effectively modeling these problems as Quadratic Unconstrained Binary Optimization (QUBO) models. The basic form of the QUBO model is presented in Eq. \ref{eq:qubo}, where \( x \) is a vector composed of 0-1 variables, and the matrix \( Q \) is the coefficient matrix, also known as the QUBO matrix
\begin{equation}
  \text{minimize}/\text{maximize} \quad y = x^{t}Qx \label{eq:qubo}.
\end{equation}
The QUBO model is structurally similar to the Ising model, which serves as the input form for the QAOA. The Hamiltonian of the Ising model is given in Eq. \ref{eq:ising}, where \( J_{ij} \) and \( h_{j} \) are coefficients, and \( \hat{\sigma}^{z} \) represents the \( z \)-component of the spin. In the Pauli representation, the matrix form of this operator is \( \hat{\sigma}^{z} = \begin{bmatrix}
  1 & 0 \\
  0 & -1
  \end{bmatrix} \). Its two eigenvalues are 1 and -1, corresponding to the eigenstates \( \ket{0} = \begin{bmatrix}
  1 \\
  0
  \end{bmatrix} \) and \( \ket{1} = \begin{bmatrix}
  0 \\
  1
  \end{bmatrix} \), respectively. These eigenstates correspond to the two possible directions of the spin operator \( \hat{\sigma} \) in the \( z \)-component.
\begin{equation}
  \hat{H}_{Ising} = \sum\limits_{i,j} J_{ij} \hat{\sigma}_{i}^{z}\hat{\sigma}_{j}^{z} + \sum\limits_{j} h_{j} \hat{\sigma}_{j}^{z}\label{eq:ising}
\end{equation}
It can be observed that the transition from Eq. \ref{eq:qubo} to Eq. \ref{eq:ising} can be achieved by first replacing the 0-1 variables with variables taking values \(\{-1,1\}\), followed by a process of operatorization. Therefore, the primary focus of this chapter will be on the QUBO modeling of the TDP. Our approach begins by modeling the TDP as a 0-1 integer programming problem. We convert the constraints of the 0-1 integer programming model into quadratic penalties and incorporate them into the original objective function, ultimately deriving the QUBO model. We start by providing a complete definition of the TDP problem. Let's recall the definition of TDP. Given a graph \( G(V, E) \), its TDS, denoted as \( D \), is a subset of \( V \) that contains no isolated vertices. For every vertex \( v \in V \setminus D \), there exists at least one vertex in \( D \) to which it is connected. The objective of the TDP is to find such a set \( D \) with the minimal size. Notably, unlike classical DP, TDP requires that there are no isolated vertices within \( D \). Next, we present the 0-1 integer programming model for the TDP.
\begin{alignat}{2}
  \min_{\{X_{i}\}} \quad & \sum\limits_{i=1}^{|V|} X_{i} \label{eq:01tdomtarget},\\
  \mbox{s.t.}\quad
  &\sum\limits_{j \in N(i)} X_{j} \ge 1 \quad \forall i \in V  \label{eq:01tdomcost},\\
  &X_{i} \in \{0,1\}  \quad \forall i \in V.
\end{alignat}
In this model, when a vertex \( i \in D \), \( X_{i} = 1 \); otherwise, \( X_{i} = 0 \). The objective function in Eq. \ref{eq:01tdomtarget} represents the size of \( D \). For each vertex \( i \), it is important to note that the constraint in Eq. \ref{eq:01tdomcost} utilizes the open neighborhood \( N(i) \) instead of the closed neighborhood \( N[i] \). In classical DP, the closed neighborhood would typically be used. However, for TDP, which requires that there are no isolated vertices in \( G[D] \), each vertex in the dominating set \( D \) must have at least one neighbor that is also in \( D \). This condition differentiates TDP from DP, making it more stringent. 

With the 0-1 integer programming model for TDP established, we now proceed to convert this model into a QUBO formulation. Since \( X_{*} \in \{0,1\} \), it follows that \( X_{*} = X_{*}^{2} \), which means the original objective function already meets the requirements of the QUBO model. Our next task is to convert the constraint in Eq. \ref{eq:01tdomcost} into quadratic penalties and incorporate these penalties into the original objective function. The general form of Constraint \ref{eq:01tdomcost} is
\begin{equation}
  X_{1} + X_{2} + \dots + X_{n} \geq 1, \quad n = |N(i)|, \quad \forall i \in V \label{eq:ctnormal}.
\end{equation}
When \( n = 1 \) or \( n = 2 \), we can transform them according to \cite{RN416} as
\begin{equation*}
  P \cdot (X_{1} - 1) ^ {2} \quad or \quad P \cdot (1- X_{1} - X_{2} + X_{1} \cdot X_{2}).
\end{equation*}
 Here, \( P \) is the punishment coefficient. When \( n \geq 3 \), we introduce a slack variable \( S \) to convert the inequality \( X_{1} + X_{2} + \dots + X_{n} \geq 1 \) into an equality constraint
\begin{equation}
  X_{1} + X_{2} + \dots + X_{n} - S - 1 = 0 \label{eq:sctdom}.
\end{equation}
Since \( X_{1} + X_{2} + \dots + X_{n} - 1 \geq 0 \) and \( X_{*} \in \{0,1\} \), it follows that \( S \in [0, n-1] \), where \( S \) can take all integer values within this range. After estimating the range of \( S \), our next task is to represent \( S \) using additional 0-1 variables. Considering the range of \( S \), we can represent \( S \) as
\begin{equation}
  S = \sum\limits_{i=1}^{bl_{n-1}-1} X_{i}^{\prime}\cdot 2^{i-1} + (n - 1 - \sum\limits_{i=1}^{bl_{n-1}-1}2^{i-1}) \cdot X_{bl_{n-1}}^{\prime} .\label{eq:sc}
\end{equation}
Here, the symbol \( bl_{n} \) denotes the length of the binary representation of the integer \( n \). For example, \( bl_{3} = 2 \) since the binary representation of 3 is \( 11 \). Additionally, \( X_{*}^{\prime} \in \{0,1\} \) represents the newly introduced variables to express \( S \). After completing these preparations, we can begin the conversion of Constraint \ref{eq:01tdomcost} into quadratic penalties. First, we can rewrite Eq. \ref{eq:sctdom} as
\begin{equation}
  P \cdot (X_{1} + X_{2} + \dots + X_{n} - S - 1) ^ {2} \label{eq:rawqp}.
\end{equation}
By substituting Eq. \ref{eq:sc} into Eq. \ref{eq:rawqp}, we obtain
\begin{equation}
  P \cdot (X_{1} + X_{2} + \dots + X_{n} - [\sum\limits_{i=1}^{bl_{n-1}-1} X_{i}^{\prime} \cdot 2^{i-1} + (n - 1 - \sum\limits_{i=1}^{bl_{n-1}-1}2^{i-1}) \cdot X_{bl_{n-1}}^{\prime}] - 1) ^ {2} \label{eq:qp}.
\end{equation}
Finally, we arrive at the QUBO model for the TDP, which can be expressed as
\begin{equation}
  \begin{split} 
  &\min_{\{X,X^{\prime}\}} \quad  \sum\limits_{i=1}^{|V|} X_{i}\\
  &+ \sum\limits_{i \in V, |N(i)| \geq 3} P \cdot [\sum\limits_{j \in N(i)}X_{j}  - (\sum\limits_{i=1}^{bl_{|N(i)|-1}-1} X_{i}^{\prime} \cdot 2^{i-1} + (|N(i)| - 1 - \sum\limits_{i=1}^{bl_{|N(i)|-1}-1}2^{i-1}) \cdot X_{bl_{|N(i)|-1}}^{\prime}) - 1]^{2}\\
  & + \sum\limits_{i \in V, |N(i)| = 2, N(i) = \{j,k\}} P \cdot (1- X_{j} - X_{k} + X_{j} \cdot X_{k})\\
  & + \sum\limits_{i \in V, |N(i)| = 1, N(i) = \{j\}} P \cdot (X_{j} - 1) ^ {2}
  \label{eq:tdomqubo}.
  \end{split}
\end{equation}
At this point, we have completed the conversion of the TDP into a QUBO model. To utilize QAOA to solve this model, we need to convert it into a Hamiltonian. First, we need to convert all 0-1 variables \( X \) and \( X^{\prime} \) into binary variables \( s \) that take values in \( \{-1, 1\} \). The conversion is
\begin{equation}
  X_{i} = \frac{s_{i} + 1}{2}.
\end{equation}
After the substitution, Eq. \ref{eq:tdomqubo} becomes
\begin{equation}
  \begin{split} 
  &\min_{\{s,s^{\prime}\}} \quad  \sum\limits_{i=1}^{|V|} \frac{s_{i} + 1}{2}\\
  &+ \sum\limits_{i \in V, |N(i)| \geq 3} P \cdot [\sum\limits_{j \in N(i)}\frac{s_{j} + 1}{2}\\  
  &- (\sum\limits_{i=1}^{bl_{|N(i)|-1}-1} \frac{s_{i}^{\prime} + 1}{2}\cdot 2^{i-1} + (|N(i)| - 1 - \sum\limits_{i=1}^{bl_{|N(i)|-1}-1}2^{i-1}) \cdot \frac{s_{bl_{|N(i)|-1}}^{\prime} + 1}{2}) - 1]^{2}\\
  & + \sum\limits_{i \in V, |N(i)| = 2, N(i) = \{j,k\}} P \cdot (1- \frac{s_{j} + 1}{2} - \frac{s_{k} + 1}{2} + \frac{s_{j} + 1}{2} \cdot \frac{s_{k} + 1}{2})\\
  & + \sum\limits_{i \in V, |N(i)| = 1, N(i) = \{j\}} P \cdot (\frac{s_{j} + 1}{2} - 1) ^ {2}.
  \label{eq:tdomqubos}
  \end{split}
\end{equation}
Next, by replacing \( s \) and \( s^{\prime} \) with the Pauli-Z operator \( \hat{\sigma}^{z} \), we can ultimately obtain the Hamiltonian
\begin{equation}
  \begin{split} 
  &\hat{H}_{c} = \sum\limits_{i=1}^{|V|} \frac{\hat{\sigma}_{i}^{z} + 1}{2}   \\
  &+ \sum\limits_{i \in V, |N(i)| \geq 3} P \cdot [\sum\limits_{j \in N(i)}\frac{\hat{\sigma}_{j}^{z} + 1}{2}\\  
  &- (\sum\limits_{i=1}^{bl_{|N(i)|-1}-1} \frac{(\hat{\sigma}_{i}^{z})^{\prime} + 1}{2} \cdot 2^{i-1} + (|N(i)| - 1 - \sum\limits_{i=1}^{bl_{|N(i)|-1}-1}2^{i-1}) \cdot \frac{(\hat{\sigma}_{bl_{|N(i)|-1}}^{z})^{\prime} + 1}{2}) - 1]^{2} \\
  & + \sum\limits_{i \in V, |N(i)| = 2, N(i) = \{j,k\}} P \cdot (1- \frac{\hat{\sigma}_{j}^{z}+ 1}{2} - \frac{\hat{\sigma}_{k}^{z} + 1}{2} + \frac{\hat{\sigma}_{j}^{z} + 1}{2} \cdot \frac{\hat{\sigma}_{k}^{z} + 1}{2})\\
  & + \sum\limits_{i \in V, |N(i)| = 1, N(i) = \{j\}} P \cdot (\frac{\hat{\sigma}_{j}^{z} + 1}{2} - 1) ^ {2}.
  \label{eq:htdom}
  \end{split}
\end{equation}
We can estimate the number of qubits required for this modeling approach.
\begin{theorem}
  The upper bound on the number of qubits required to model TDP using the method presented in this paper is \(2|V| + |V| \log_{2}\left( \frac{2|E|}{|V|} - 1 \right)\).
\end{theorem}
\begin{proof}
  For an arbitrary graph \( G = (V, E) \), the number of qubits required to solve TDP on this graph arises from two components. The first component consists of \( |V| \) qubits, which are needed to map each vertex. The second component comprises the additional qubits that may be introduced when converting the total condition into a quadratic penalty term. According to Eq. \ref{eq:ctnormal}, no additional qubits are required when \( |N(i)| = 1 \) or \( 2 \); when \( |N(i)| \geq 3 \), based on Eq. \ref{eq:sc}, the number of additional qubits needed is \( bl_{n-1} \), where \( n = |N(i)| \). As per the definition in this paper, it is straightforward to observe that \( bl_{n} = \lfloor \log_{2} n \rfloor + 1 \). Let \( d_{i} \) denote the degree of vertex \( i \), such that \( d_{i} = |N(i)| \) for all \( i \in V \). The total number of qubits will not exceed
  \begin{align}
    & |V| + \sum\limits_{i \in V} (\lfloor \log_{2} (d_{i} - 1) \rfloor + 1)\\
   \leq & 2|V| + \sum\limits_{i \in V} \lfloor \log_{2} (d_{i} - 1) \rfloor \\
   \leq & 2|V| + \sum\limits_{i \in V} \log_{2} (d_{i} - 1) \\
   = & 2|V| + \frac{1}{|V|}|V| \log_{2} \prod_{i \in V} (d_{i} - 1) \\
   = & 2|V| + |V| \log_{2} (\prod_{i \in V} (d_{i} - 1))^{\frac{1}{|V|}} \\
   \leq & 2|V| + |V| \log_{2} \frac{\sum\limits_{i \in V}(d_{i} - 1)}{|V|}\\
   = & 2|V| + |V| \log_{2} \frac{2|E| - |V|}{|V|} \\
   = & 2|V| + |V| \log_{2}(\frac{2|E|}{|V|} - 1).
  \end{align}
\end{proof}

In addition to deriving an upper bound on the number of qubits required for TDP, we can also report an interesting observation: for the same graph \( G = (V, E) \), TDP actually requires fewer qubits than DP. Let \( g \) denote the reduction in the number of qubits required for TDP compared with DP. Partition the vertex set as \( V = V_{0} \cup V_{1} \cup V_{2} \cup V_{\ge 3} \), where \( V_{j} := \{ i \mid i \in V \text{ and } d_{i} = j \} \) and \( V_{\ge j} := \{ i \mid i \in V \text{ and } d_{i} \ge j \} \). Using this notation, we can state and prove Theorem \ref{theorem:gap}.  

\begin{theorem} \label{theorem:gap}
     $ 2|V_2| \leq g \le 2|V_2| + |V_{\ge 3}|$.
\end{theorem}

\begin{proof}
Based on the preceding discussion, the number of qubits required for TDP, denoted by \( q_{\text{tdp}} \), can be readily expressed as  
\begin{equation}
    q_{\text{tdp}} = |V| + 0 \cdot |V_{0}| + 0 \cdot |V_{1}| + 0 \cdot |V_{2}| + \sum\limits_{i \in V_{\ge 3}}(\lfloor log_2 (d_i - 1)  \rfloor + 1).
\end{equation}
It should be emphasized that the first term, \(|V|\), counts the qubits corresponding to the vertices, while the remaining terms account for the auxiliary qubits introduced solely for modeling purposes. Similarly, based on the comparative analysis of DP and TDP in Section \ref{sec:problemmodeling}, the number of qubits required for DP, denoted by \(q_{\text{dp}}\), can be written as  
\begin{equation}
    q_{\text{dp}} = |V| + 0 \cdot |V_{0}| + 0 \cdot |V_{1}| + 2 \cdot |V_{2}| + \sum\limits_{i \in V_{\ge 3}}(\lfloor log_2 (d_i)  \rfloor + 1).
\end{equation}
Then
\begin{equation}
    g = q_{\text{dp}} - q_{\text{tdp}} = 2|V_{2}| +  \sum\limits_{i \in V_{\ge 3}}(\lfloor log_2 (d_i)  \rfloor - \lfloor log_2 (d_i - 1)  \rfloor),
\end{equation}
and 
\begin{equation*}
    \lfloor log_2 (d_i)  \rfloor - \lfloor log_2 (d_i - 1)  \rfloor \in \{0,1\}, \forall \, d_{i}\ge3, 
\end{equation*}
then we have 
\begin{equation*}
    2|V_2| \leq g \le 2|V_2| + |V_{\ge 3}|.
\end{equation*}
\end{proof}

The validity of Theorem~\ref{theorem:gap} implies that when solving DP and TDP on the same graph, TDP requires fewer qubits. Given that we are currently in the Noisy Intermediate-Scale Quantum (NISQ) era, where qubit resources remain scarce, initiating investigations with TDP is arguably the more appropriate choice.  

At this point, we have completed the process of converting TDP into a Hamiltonian, thereby finalizing the preparations for solving TDP using QAOA. In the subsequent chapters, we will introduce the basic concepts of QAOA.  

\section{QAOA}\label{sec:QAOA}  

We illustrate the basic concepts of the QAOA with reference to Fig. \ref{fig:qaoaflow}. For a quantum system composed of \( n \) qubits, the state vector of the system is represented by the bit string \( z \), where \( z = z_{1}z_{2}z_{3}\dots z_{n} \) and \( z_{i} \in \{0,1\} \), corresponding to the two distinct spin orientations. The system Hamiltonian, \( H_{c} \), was derived from the QUBO model of TDP in the preceding steps. This establishes a mapping from the objective function of the combinatorial optimization problem to the energy of the quantum system. The QAOA algorithm initiates by preparing the initial state \( \ket{s} \) from \( \ket{\underbrace{00\dots 0}_{n}} \) via a Hadamard gate.  

\begin{figure}[H]
  \centering
  \includegraphics[width=12cm]{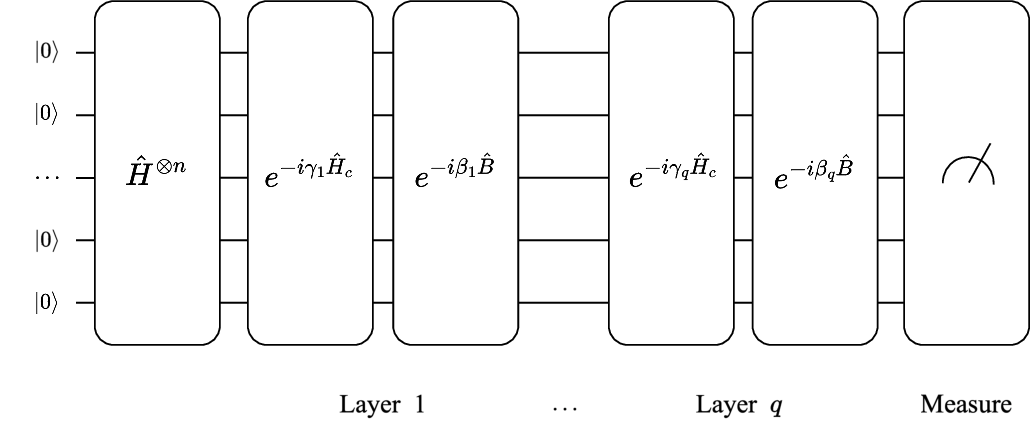}
  \caption {Working flow of QAOA with $q$ layers and $n$ qubits}
  \label{fig:qaoaflow}
\end{figure}
\noindent The definition of the Hadamard gate is
\begin{equation}
    \hat{H} = 2^{-\frac{1}{2}} ([\ket{0} + \ket{1}] \bra{0} + [\ket{0} -\ket{1}] \bra{1}).
\end{equation}
Applying the Hadamard gate to the state \(\ket{0}\) will prepare it in an equal probability superposition of \(\ket{0}\) and \(\ket{1}\).
\begin{equation}
  \hat{H} \ket{0} = 2^{-\frac{1}{2}} [\ket{0} + \ket{1}].
\end{equation}
Apply the Hadamard gate to each qubit,
\begin{equation}
  \ket{s} = \underbrace{\hat{H} \otimes \hat{H} \cdots \otimes \hat{H}}_{n} \ket{\underbrace{00\dots 0}_{n}} = \frac{1}{\sqrt{2^{n}} } \cdot \sum\limits_{z} \ket{z}.
\end{equation}
In the initial state \( \ket{s} \), if we measure the wave function of the quantum system directly, the probability of the system collapsing into any state vector is \( \frac{1}{\sqrt{2^{n}}} \). QAOA applies two types of unitary transformations, \( U(\hat{C},\gamma) \) and \( U(\hat{B},\beta) \) (Eq. \ref{eq:uc}, Eq. \ref{eq:ub}), to the initial state \( \ket{s} \), repeating this process \( q \) times. Here, \( \hat{C} = \hat{H}_c \), \( \hat{B} = \sum\limits_{j=1}^{n} \hat{\sigma}_{j}^{x} \), \( \gamma \in [0,2\pi] \), and \( \beta \in [0,\pi] \).
\begin{align}
  U(\hat{C},\gamma) &= e^{-i\gamma \hat{H}_{c}} \label{eq:uc}\\
  U(\hat{B},\beta) &=  e^{-i\beta \hat{B}} \label{eq:ub}
\end{align}
In Fig. \ref{fig:qaoaflow}, \( \gamma_{q} \) and \( \beta_{q} \) represent the two types of angles at the \( q \)-th layer. After repeating this process \( q \) times, we obtain the final state \( \ket{\gamma,\beta} \).
\begin{equation}
  \ket{\gamma,\beta} = U(\hat{B},\beta_{q})U(\hat{C},\gamma_{q}) \cdots U(\hat{B},\beta_{1})U(\hat{C},\gamma_{1}) \ket{s} \label{{eq:gammabeta}}
\end{equation}
We can obtain the expected value of \( \hat{H}_{c} \), denoted as \( F_{q}(\gamma,\beta) \), by repeatedly loading the circuit and measuring \( \ket{\gamma,\beta} \) multiple times.
\begin{equation}
  F_{q}(\gamma,\beta) = \bra{\gamma,\beta} \hat{H}_{c} \ket{\gamma,\beta}
\end{equation}
Since TDP is a minimization problem, we need to continuously adjust the angles \( \gamma \) and \( \beta \) at each layer to minimize \( F_{q}(\gamma,\beta) \). QAOA is a hybrid algorithm that integrates quantum computing with classical optimization, typically employing classical optimizers such as COBYLA to tune the angles at each layer. The optimization process terminates when the maximum number of iterations is reached or a specified function tolerance is satisfied. At this stage, we obtain the optimal angles \( \gamma_{*} \) and \( \beta_{*} \), which are then used to update the quantum circuit. Multiple samplings are performed, and the bit string \( z_{*} \) with the highest probability from the sampling results is outputted. The final TDS can be extracted from \( z_{*} \).  
\section{Experiment}\label{sec:Experiment}

The experimental environment utilizes an AMD R9 7950X3D CPU with 48 GB of memory. We employed IBM's Qiskit package to construct the QAOA quantum circuits, perform backend simulations, and conduct sampling. The optimizer used is COBYLA from the Scipy optimize package, with a default function tolerance of \( 10^{-8} \). For the initial values of \( \gamma \) and \( \beta \) in QAOA, we adopted the initialization method proposed in \cite{RN458}. The graph used in this study is depicted in Fig. \ref{fig:6nodegraphtdom}. This graph consists of six vertices, and the size of its minimal DS is 2, which can include \{2,5\}, \{0,2\}, \{1,4\}, or \{4,0\}. However, none of these sets meet the conditions for a TDS since all of them contain isolated vertices. Similar to minimal DS, there are multiple minimal TDS, such as \{0,1,2\}, \{0,4,5\}, \{1,2,4\}, and \{2,4,5\}. Considering the different sizes of minimal DS and minimal TDS, this graph effectively helps us validate the performance of QAOA in solving the TDP rather than DP.
\begin{figure}[htbp]
    \centering
    \begin{tikzpicture}[scale=1.2,
    every node/.style={circle, draw, minimum size=1em, inner sep=1pt},
    every edge/.style={draw, thick}]
  \node (0) at (0,1)  {0};
  \node (1) at (1,2)  {1};
  \node (2) at (2,2){2};
  \node (3) at (3,1){3};
  \node (4) at (2,0){4};
  \node (5) at (1,0){5};
  \draw (0) edge (1)
         (0) edge (5)
        (1) edge (2)
        (2) edge (3)
        (3) edge (4)
        (4) edge (5)
        (4) edge (2);
\end{tikzpicture}
\caption{A graph with 6 nodes and 7 edges}
    \label{fig:6nodegraphtdom}
\end{figure}
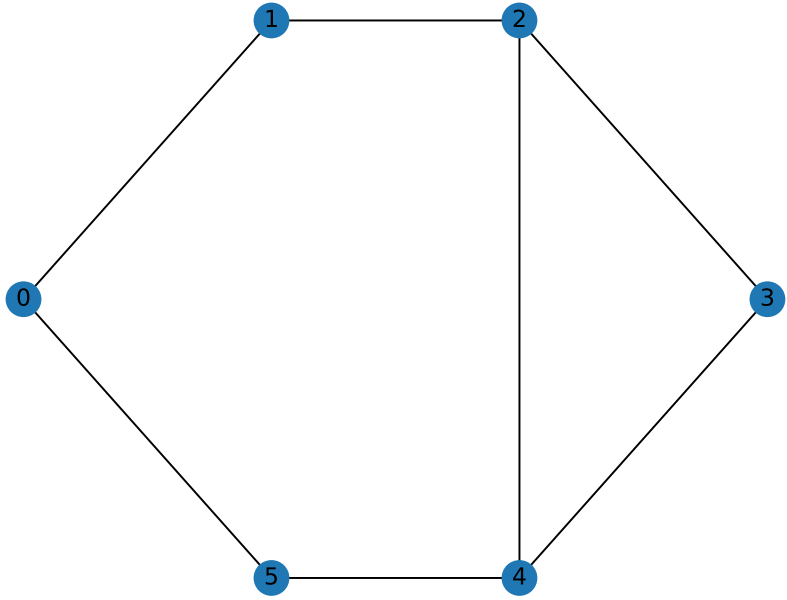

Additionally, based on the conversion method introduced in Section \ref{sec:problemmodeling}, the QUBO model for the TDP of this graph is given by Eq. \ref{eq:expqubo}. This model can be easily transformed into the Hamiltonian by substituting \( x_{*} \) with \( \frac{s_{*} + 1}{2} \) and replacing \( s_{*} \) with \( \hat{\sigma}_{*}^{z} \). For the sake of brevity, we will not expand on this here.
\begin{equation}
  \begin{split}
    \text{minimize} \quad & x_0 + x_1 + x_2 + x_3 + x_4 + x_5 \\
    &+ P \cdot (1 - x_1 - x_5 + x_1 \cdot x_5) \\
    &+ P \cdot (1 - x_0 - x_2 + x_0 \cdot x_2) \\
    & + P \cdot (x_1 + x_3 + x_4 - (x_6 + x_7) - 1)^2 \\
    & + P \cdot (1 - x_2 - x_4 + x_2 \cdot x_4) \\
    & + P \cdot (x_2 + x_3 + x_5 - (x_8 + x_9) - 1)^2 \\
    & + P \cdot (1 - x_0 - x_4 + x_0 \cdot x_4), \quad x_{*} \in \{0,1\} \label{eq:expqubo}
  \end{split}
\end{equation}
The parameters involved in the entire experiment include the layer number \( q \), the punishment coefficient \( P \), and the maximum iterations. To ensure the completeness of the testing, we evaluated the performance of QAOA in solving the TDP across a total of 128 parameter combinations, with \( q \in \{2, 5, 10, 20\} \), \( P \in \{4.8, 5.4, 6.0, 6.6, 7.2, 7.8, 8.4, 9.0\} \), and maximum iterations \( \in \{50, 100, 200, 500\} \). The values of \( P \) correspond to 0.8, 0.9, 1.0, 1.1, 1.2, 1.3, 1.4, and 1.5 times the total number of vertices. This choice of values is inspired by \cite{RN416}.

First, we present some experimental results obtained under specific parameters. We set \( q = 5 \), \( P = 9.0 \), and maximum iterations = 500. Fig. \ref{fig:expdisdom} shows the probability distribution of bit strings from the final sampling results, with the bit string \( z = 100011 \) highlighted in purple due to its highest sampling probability of 0.0655. This bit string's sampling probability is noticeably higher than that of the others. By marking the corresponding vertex set \( \{0,4,5\} \) in the graph (Fig. \ref{fig:d1optdom}), we confirm that this bit string corresponds to a DS that meets the TDS conditions, and that this TDS is minimal. Additionally, in Fig. \ref{fig:expdisdom}, we observe that the sampling probability for \( z = 111000 \) is 0.0653, making it the second highest probability bit string. Through visualization (Fig. \ref{fig:d2optdom}), we see that the corresponding set \( \{0,1,2\} \) forms another minimal TDS. Although our examination of the remaining bit strings in Fig. \ref{fig:expdisdom} did not reveal significantly higher probabilities for the two other potential TDS, \( \{1,2,4\} \) and \( \{2,4,5\} \), we can still affirm that QAOA is effective for solving the TDP under these parameters. In Fig. \ref{fig:costtdom}, we record the trend of cost variation throughout the iterations. We find that, between iterations \( \in [0, 60] \), the cost fluctuates significantly. Once the number of iterations exceeds 60, the cost changes become more stable. This indicates that the QAOA exhibits good convergence when solving the TDP with the current parameters. We believe that the dramatic fluctuations in cost are due to the penalty terms in the Hamiltonian. Since the penalty must exceed the original objective function value, when COBYLA iterates through the angles, it faces situations where some penalty terms become zero while others back to non-zero. This causes sharp variations in cost. However, once most penalties are satisfied and reduced to zero, the cost becomes smoother. This transition signifies a shift in the optimization focus from finding a TDS to identifying a smaller TDS.  

\begin{figure}[H]
  \centering
  \includegraphics[width=12cm]{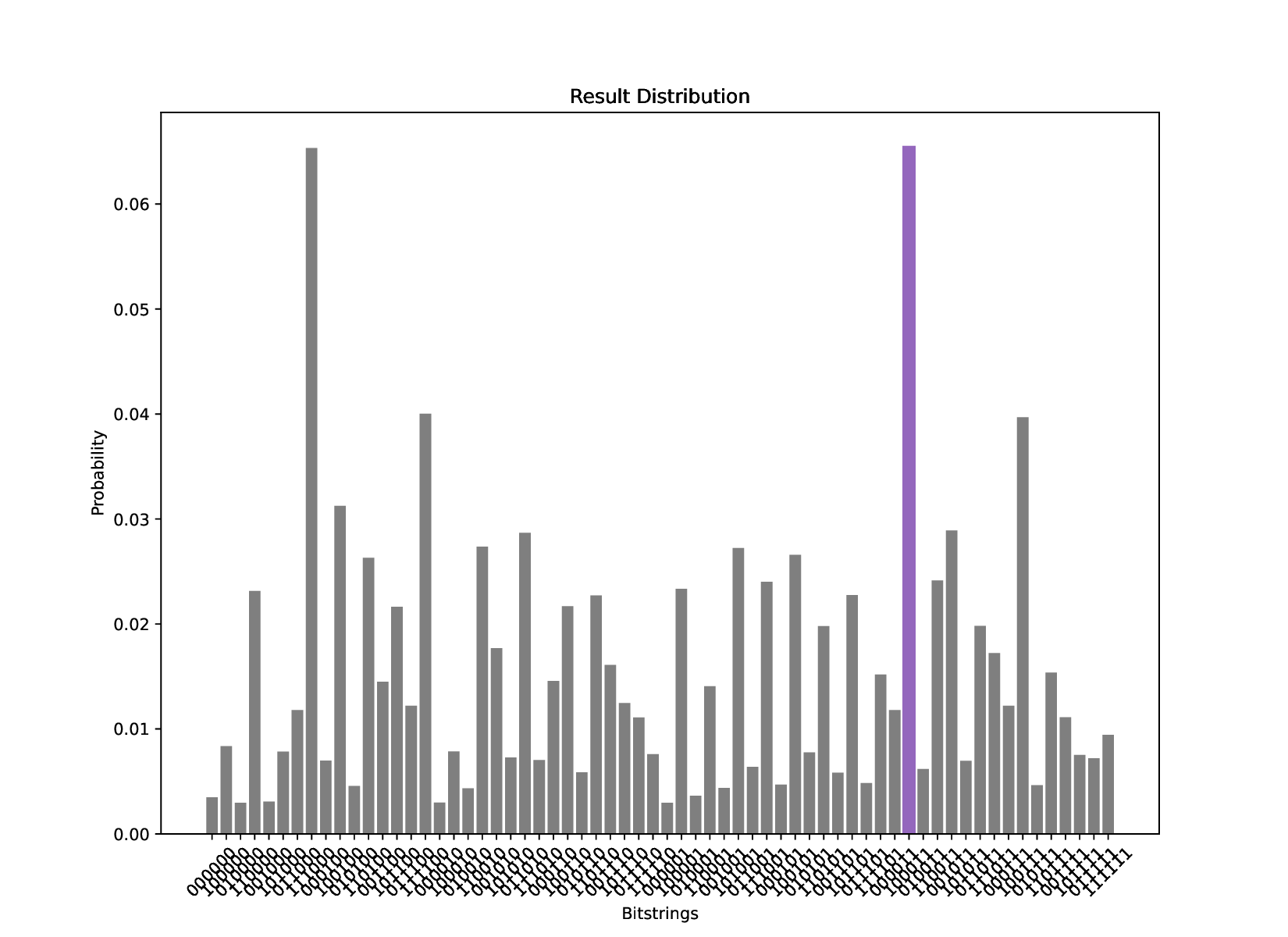}
  \caption {Probability distribution of the final sampling result when $q = 5$, $P = 9.0$, and maximal iterations = 500. The sampling probability for the most probable bit string is highlighted in purple.}
  \label{fig:expdisdom}
\end{figure}


\begin{figure}[htbp]
  \centering
  \begin{tikzpicture}[scale=1.2,
    every node/.style={circle, draw, minimum size=1em, inner sep=1pt},
    every edge/.style={draw, thick}]

    \node(1) at (1,2) {1};
    \node(2) at (2,2) {2};
    \node(3) at (3,1) {3};
    \node[red, fill=red!20] (0) at (0,1) {0};
    \node[red, fill=red!20] (4) at (2,0) {4};
    \node[red, fill=red!20] (5) at (1,0) {5};
    \draw (0) edge (1)
          (0) edge (5)
          (1) edge (2)
          (2) edge (3)
          (3) edge (4)
          (4) edge (5)
          (4) edge (2);
  \end{tikzpicture}
  \caption {Visualization of the bit string 100011 where the vertices in TDS are marked in red}
  \label{fig:d1optdom}
\end{figure}

\begin{figure}[htbp]
  \centering
  \begin{tikzpicture}[scale=1.2,
    every node/.style={circle, draw, minimum size=1em, inner sep=1pt},
    every edge/.style={draw, thick}]
    \node[red, fill=red!20](1) at (1,2) {1};
    \node[red, fill=red!20](2) at (2,2) {2};
    \node(3) at (3,1) {3};
    \node[red, fill=red!20](0) at (0,1) {0};
    \node(4) at (2,0) {4};
    \node(5) at (1,0) {5};
    \draw (0) edge (1)
          (0) edge (5)
          (1) edge (2)
          (2) edge (3)
          (3) edge (4)
          (4) edge (5)
          (4) edge (2);
  \end{tikzpicture}
  \caption {Visualization of the bit string 111000 where the vertices in TDS are marked in red}
  \label{fig:d2optdom}
\end{figure}


\begin{figure}[htbp]
  \centering
  \includegraphics[width=11cm]{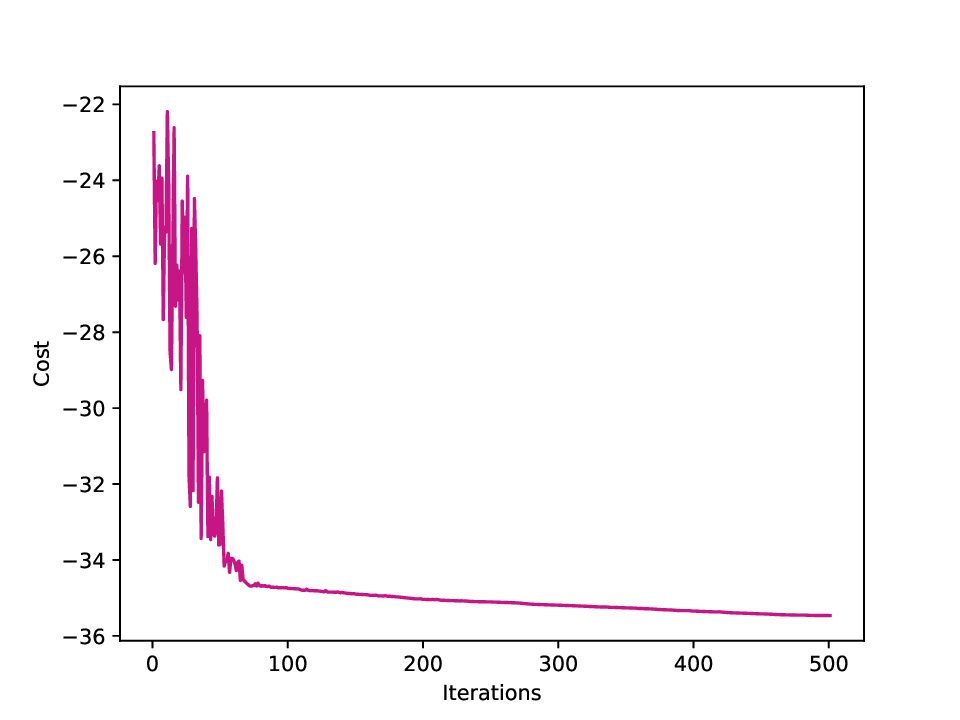}
  \caption {Cost of the QAOA when $q = 5$, $P = 9.0$, and maximal iterations = 500}
  \label{fig:costtdom}
\end{figure}

\begin{figure}[htbp]
  \centering
  \includegraphics[width=11cm]{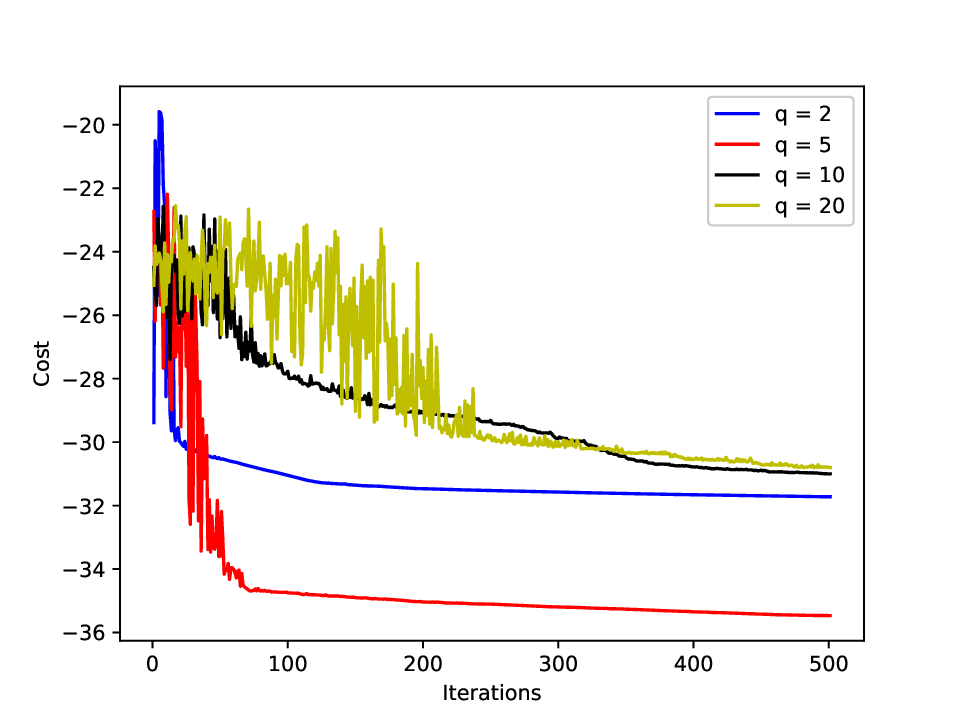}
  \caption {The comparison of cost for different $q$ when $P = 9.0$ and maximal iterations = 500} 
  \label{fig:tdomcostcompare}
\end{figure}

Next, we will present a comparison of QAOA computational results under different parameter configurations. First, in Fig. \ref{fig:tdomcostcompare}, we illustrate the cost variations for different values of \( q \). We observe that when the number of layers \( q \) is larger, such as \( q = 10 \) or \( 20 \), the amplitude and range of cost fluctuations are greater compared to \( q = 2 \) and \( 5 \). Even when the number of iterations reaches its maximum, the costs for \( q = 10 \) and \( 20 \) continue to fluctuate. We believe this is due to a significant number of penalty terms remaining non-zero after the initial fluctuations have ended. Additionally, the increased number of variables associated with more layers may lead the COBYLA algorithm to prematurely converge on local optima. We can confirm that, for the current instance we are using (Fig. \ref{fig:6nodegraphtdom}), setting too many layers for QAOA is ineffective and may require more iterations. For \( q = 2 \) and \( q = 5 \), the cost clearly settles into a local optimum when \( q = 2 \). According to \cite{RN436}, we can derive Eq. \ref{eq:jx}, which indicates that as the number of layers \( q \) increases, the expected value of the sampled \( \hat{H}_c \) approaches the size of the minimal TDS. We believe that in the current instance, the smaller values of \( q \), such as \( 2 \) and \( 5 \), align with this trend. However, the number of layers, the chosen optimization algorithm, maximum iterations, punishment coefficient \( P \), and other parameters may all limit the capabilities of the QAOA designed for larger layer structures. These bottlenecks could ultimately prevent the QAOA from finding better values for \( \gamma \) and \( \beta \) within the specified range of maximum iterations.

\begin{equation}
   \lim\limits_{q \to \infty} \min \limits_{\gamma,\beta}F_{q}(\gamma,\beta)  = \min |TDS| \label{eq:jx}
\end{equation}

\begin{figure}[H]
  \centering
  \includegraphics[width=12cm]{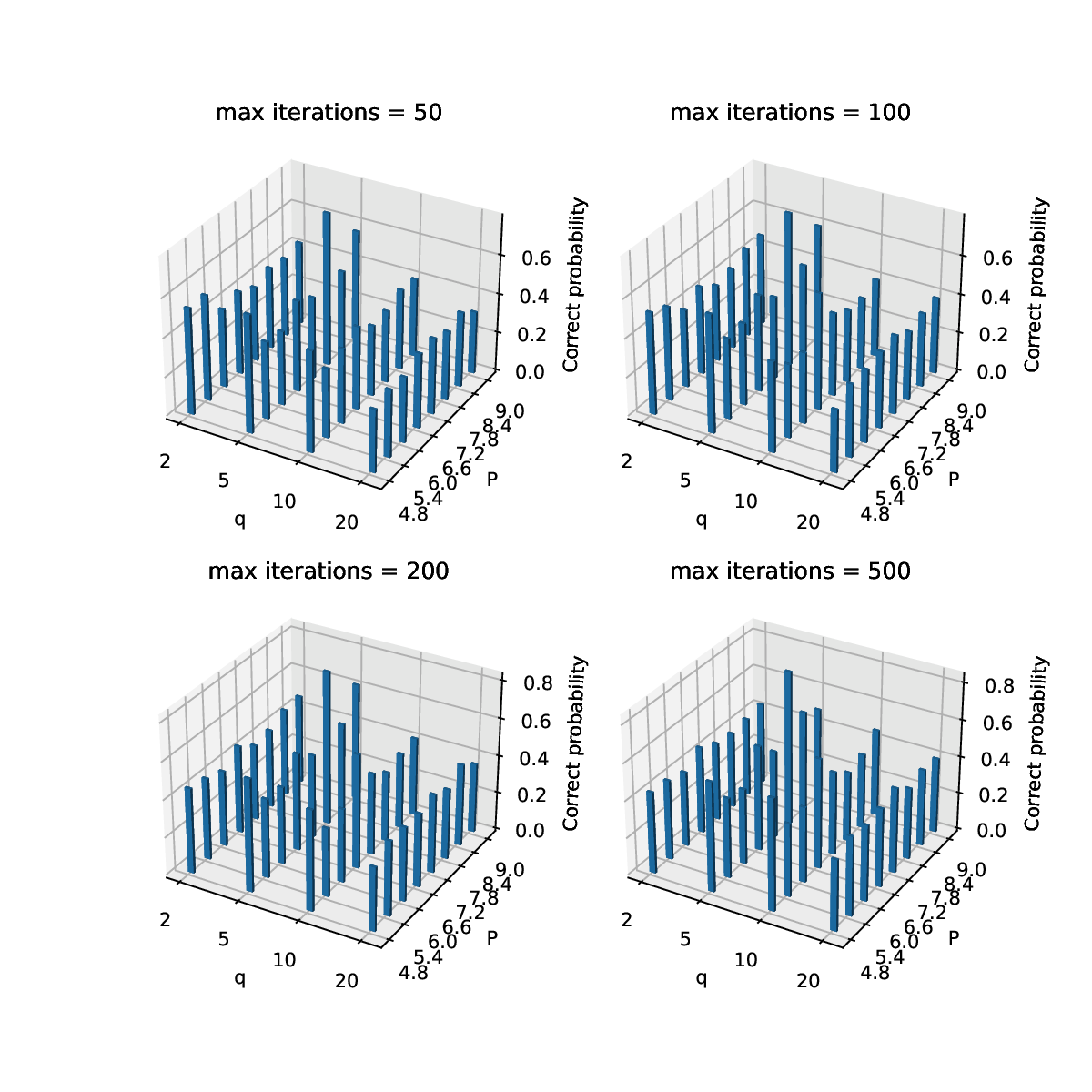}
  \caption {Correct probabilities for $q \in \{2, 5, 10, 20\}$, $P \in \{4.8,5.4,6.6,6.6,7.2,7.8,8.4,9.0\}$ and maximal iterations $\in \{50, 100, 200, 500\}$}
  \label{fig:tdomcpcompare}
\end{figure}
\begin{figure}[H]
  \centering
  \includegraphics[width=12cm]{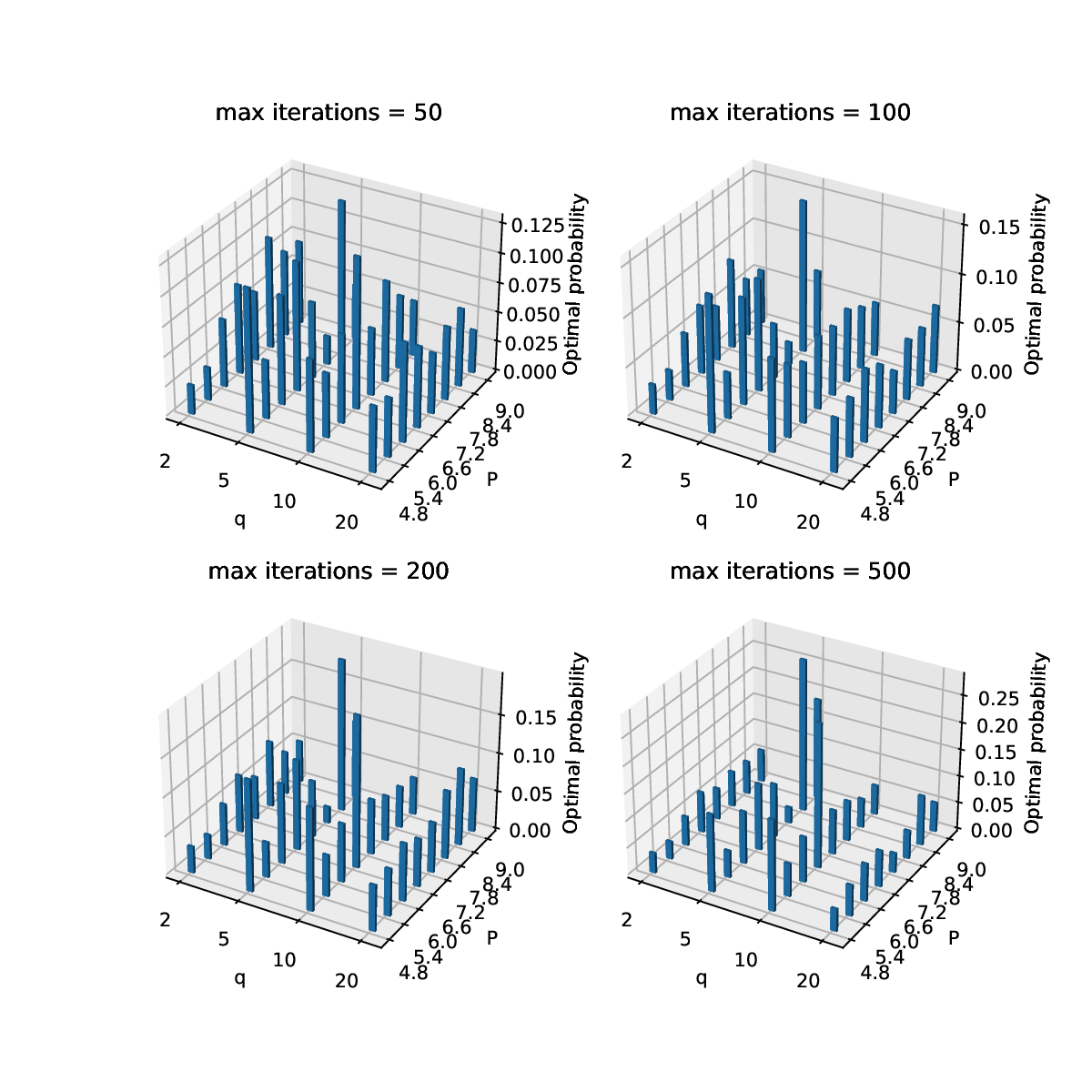}
  \caption {Optimal probabilities for $q \in \{2, 5, 10, 20\}$, $P \in \{4.8,5.4,6.6,6.6,7.2,7.8,8.4,9.0\}$ and maximal iterations $\in \{50, 100, 200, 500\}$}
  \label{fig:tdomopcompare}
\end{figure}
Next, in Figs. \ref{fig:tdomcpcompare} and \ref{fig:tdomopcompare}, we present the total probabilities of obtaining TDS (Correct probability) and minimal TDS (Optimal probability) from all sampled bit strings under different parameter combinations. Overall, we find that within the parameter ranges we set, the probability of QAOA successfully identifying TDS remains fairly stable, generally within the range of [0.3, 0.7]. Higher maximum iterations contribute to increasing the upper probability limit for finding TDS, leading us to conclude that the QAOA can relatively easily identify TDS. Regarding the optimal TDS, we observe that higher probabilities tend to cluster in regions with fewer layers, a trend that holds consistently across different levels of maximum iterations. Furthermore, among the parameter combinations yielding the highest probabilities, we notice that the punishment coefficients are generally larger. One possible explanation for this is that within a certain range, a larger punishment coefficient leads to more pronounced fluctuations in cost during the initial phases, which may assist classical optimization algorithms in escaping local optima during the optimization process.

The results shown in Figs. \ref{fig:distds} and \ref{fig:dismtds} further validate this observation. In Fig. \ref{fig:distds}, we record the distribution of parameter points where \( z_* \) corresponds to TDS, while in Fig. \ref{fig:dismtds}, we document the distribution for minimal TDS. We find that out of 128 parameter combinations, 93 yielded TDS, with 12 resulting in minimal TDS. Additionally, the distribution of parameter points for TDS is quite uniform, whereas the distribution for minimal TDS exhibits a clear tendency, clustering in regions characterized by smaller \( q \) values and larger \( P \) values. This aligns with our earlier analysis. Based on these characteristics, we conclude that when using QAOA to solve TDP, it is quite efficient in finding a TDS and is relatively insensitive to parameter choices. However, to identify the optimal TDS, more careful selection of parameters is necessary.
\begin{figure}[H]
  \centering
  \includegraphics[width=12cm]{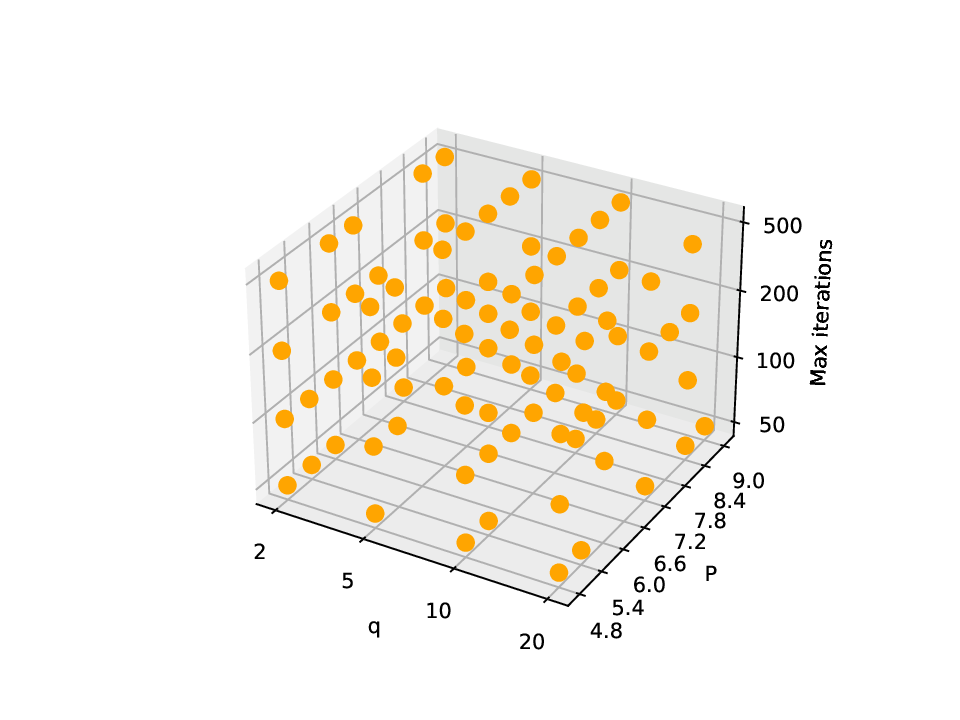}
  \caption {Parameter points for $z_{*}$ which is a TDS}
  \label{fig:distds}
\end{figure}

\begin{figure}[H]
  \centering
  \includegraphics[width=12cm]{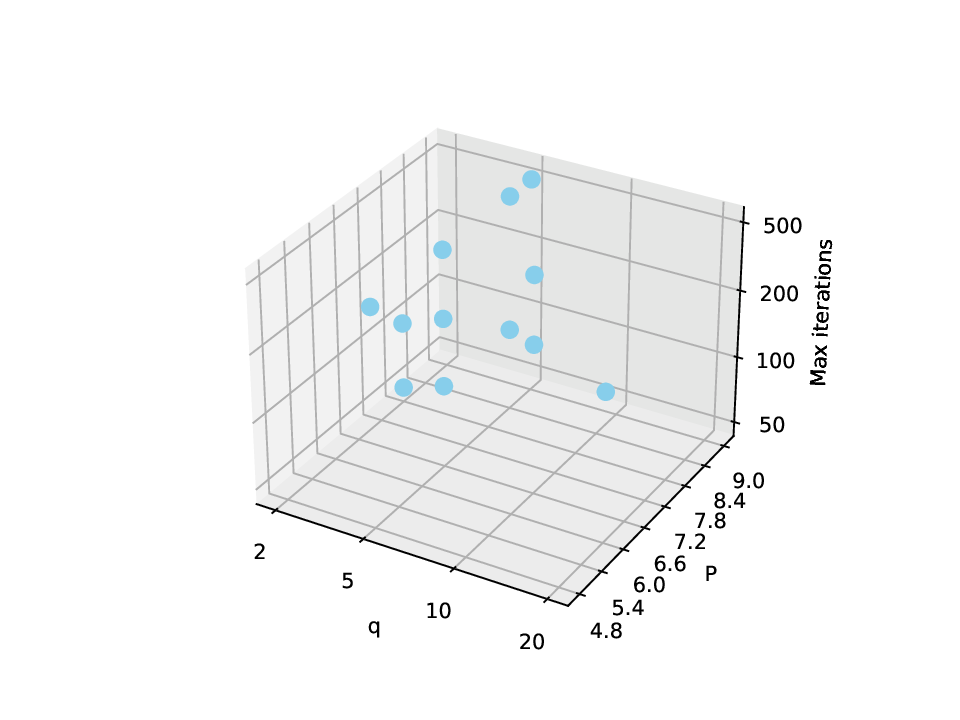}
  \caption {Parameter points for $z_{*}$ which is a minimal TDS}
  \label{fig:dismtds}
\end{figure}
Additionally, in Table \ref{table:diffcom}, we compare the time complexity of QAOA with other algorithms for solving TDP. Here, \( \mathscr{M} \) represents the maximal iterations of COBYLA, \( n \) denotes the number of vertices in the graph, and \( m \) refers to the number of edges. The time complexity analysis of QAOA has already been provided in \cite{RN334}. According to this analysis, the time complexity of QAOA consists of two components. The first component is \( O[poly(q)] \) \citep{RN480}, which is related to the \( q \)-layer structure of the algorithm. The second component is associated with the time complexity of the classical optimization algorithm used. Since we employed COBYLA, its time complexity is \( O[poly(\mathscr{M})] \). Therefore, the overall time complexity is \( O[poly(q) + poly(\mathscr{M})] \). Comparing this with the algorithms listed in Table \ref{table:diffcom}, we find that the time complexity of QAOA is acceptable. Although the algorithm may not currently have a clear advantage for specific classes of graphs, it is designed to be general-purpose, which provides an edge over other algorithms in terms of versatility. Moreover, this study did not optimize the quantum circuits, instead using those generated directly by IBM’s Qiskit. Additionally, we did not tailor or explore other classical optimization algorithms specific to the characteristics of the TDP. These factors indicate that there is significant potential for further development in this work.

\begin{longtable}{|p{2.8cm}|p{4cm}|p{2.8cm}|}
  \caption{The comparison of time complexity of algorithms for TDP}
  \label{table:diffcom}\\
  \hline
   Algorithm & Time complexity & For specific graphs?\\
  \hline
  \endfirsthead
  \multicolumn{3}{r}{Continued}\\
  \hline
  Algorithm & Time complexity & For specific graphs? \\
  \hline
  \endhead
  \hline
  \multicolumn{3}{r}{Continued on next page}\\
  \endfoot
  \endlastfoot
  \hline
  QAOA & $O[poly(q) + poly(\mathscr{M})]$ & No\\
  \hline
  \cite{RN469} & $O(n \ln(n))$ & Interval\\
  \hline
  \cite{RN477}& $O(nm^{2})$ & Cocomparability\\
  \hline
  \cite{RN468}& $O(nm^{2})$ & Cocomparability\\
  \hline
  \cite{RN478} & $O(n^{6})$ & Asteroidal 
  triple-free\\
  \hline
  \cite{RN479} & linear time & Distance 
  hereditary\\
  \hline
  \cite{RN474} & linear time & Tree\\
  \hline
\end{longtable}

\section{Conclusion}\label{sec:conclusion}

This paper investigates the use of QAOA to solve the TDP. In the modeling section, We first model the TDP as a 0-1 integer programming problem, then convert the constraints into quadratic penalties and integrate them into the original objective function, resulting in the QUBO formulation for TDP. By further transforming the QUBO model into a Hamiltonian, we complete the preparation for solving the TDP using QAOA. Also in this section, we rigorously derive the upper bound of the number of qubits required to solve the TDP, and compare it with the number of qubits needed for the classical DP, obtaining an interval range for the gap between them. In the numerical experiment section, We conducted detailed tests of QAOA's performance on TDP using a quantum simulator across 128 parameter combinations. The results indicate that QAOA successfully computes correct TDS for most parameter settings, and in approximately 10\% of cases, it outputs the optimal TDS. This suggests that the ability of QAOA to find optimal TDS is dependent on the parameters used. Additionally, an analysis of the bit string probability distribution in the final sampling results under varying parameters reveals that the accuracy of QAOA remains relatively stable within our selected parameter range, while optimal probabilities tend to favor smaller layer numbers and larger punishment coefficients. Based on the current experimental findings, we conclude that using QAOA to solve TDP is feasible and has significant potential for further development, although a more detailed parameter analysis is necessary to enhance the likelihood of obtaining optimal TDS.

Based on the experimental results and analysis, we identify the limitations of this work as follows: (1) We did not optimize our quantum circuits according to the specific characteristics of the TDP; the circuits used were generated by IBM's Qiskit. (2) The classical optimization algorithm employed could be further refined to minimize the required maximal iterations. (3) We did not conduct practical tests on a quantum computer.

Finally, we believe that the work presented in this paper can be extended in the following directions: (1) Conducting tests on quantum computers using QAOA or other quantum algorithms to evaluate the performance in a broader parameter range and across more test instances for solving the TDP. (2) Further enhancing the probability of QAOA successfully identifying the optimal TDS through detailed parameter range analysis, quantum circuit optimization, and exploring additional classical optimization algorithms. (3) Building on this work to apply the QAOA to solve other variants of the domination problem, such as perfect domination and k-domination.

\section*{Acknowledgement}

The corresponding author would like to acknowledge the support provided by China Communications Information \& Technology Group Co., Ltd. for this work.

\section*{Data Availability}
The data used to support the findings of this study are included within the article.

\section*{Conflicts of interest}
The authors declare that they have no conflicts of interest that could have appeared to influence the work reported in this paper.

\section*{Funding statement}
 
This work is supported by Fundamental Research Funds for the Central Universities and National Natural Science Foundation of China (No. 12331014).

\section*{Declaration of Generative AI}

During the preparation of this work the authors used chatgpt in order to improve readability and language. After using this tool, the authors reviewed and edited the content as needed and take full responsibility for the content of the publication.

 \bibliographystyle{elsarticle-harv} 
 \bibliography{tdom}

\end{CJK}
\end{document}